\documentclass[submission,copyright,creativecommons]{eptcs}
\providecommand{\event}{WOC 2015} 
\usepackage{breakurl}             
\usepackage{amsmath}
\usepackage{amssymb}
\usepackage{amsthm}
\usepackage{stmaryrd}
\usepackage{mathpartir}
\usepackage{preamble}
\usepackage{txfonts}
\usepackage{todonotes}
\usepackage{tikz-cd}
\usepackage{paralist}

\long\def\omitnow#1{\relax}

\newcommand{\dmapsto}{\mapsto\!\!\!\!\!\!\!\to}

\newcommand{\dhreduct}{\hreduct\!\!\!\!\!\!\to}

\newcommand{\red}{\rightarrow}

\newcommand{\reds}{\rightarrow\!\!\!\!\!\!\!\to}

\title{First Class Call Stacks: Exploring Head Reduction}
\author{Philip Johnson-Freyd, Paul Downen and Zena M. Ariola
\institute{University of Oregon\\
Oregon, USA}
\email{\{philipjf,pdownen,ariola\}@cs.uoregon.edu}
}
\def\titlerunning{First Class Call Stacks}
\def\authorrunning{P Johnson-Freyd, P. Downen \& Z.M. Ariola}
\begin{document}
\maketitle

\begin{abstract}
Weak-head normalization is inconsistent with functional extensionality
in the call-by-name $\lambda$-calculus.  We explore this problem from a new
angle via the conflict between extensionality and effects. Leveraging 
ideas from work on the $\lambda$-calculus with control, we derive
and justify alternative operational semantics and a sequence
of abstract machines for performing head reduction.
Head reduction avoids the problems with weak-head reduction and extensionality,
while our operational semantics and associated abstract machines show us how
to retain weak-head reduction's ease of implementation.
\end{abstract}

\section{Introduction}
Programming language designers are faced with multiple, sometimes
contradictory, goals.  On the one hand, it is important that users be
able to reason about their programs.  On the other hand, we want our
languages to support simple and efficient implementations.  For the
first goal, \emph{extensionality} is a particularly desirable property.  We
should be able to use a program without knowing how it was written,
only how it behaves.  In the $\lambda$-calculus, extensional reasoning
is partially captured by the $\eta$ law, a strong equational property
about functions which says that functional delegation is unobservable:
$\lambda{x}.f~x=f$.  It is essential to many proofs about functional
programs: for example, the well-known ``state monad'' only obeys the
monad laws if $\eta$ holds~\cite{marlow:statemonad:02}.  The $\eta$
law is compelling because it gives the ``maximal'' extensionality
possible in the untyped $\lambda$-calculus: if we attempt to equate
any additional terms beyond $\beta$, $\eta$, and $\alpha$, the theory
will collapse~\cite{bohm1968alcune}.
For the second goal, it is important to have normal forms,
specifying the possible results of execution, that can be
efficiently computed. To that end, most implementations 
stop execution when they encounter a lambda-abstraction. This is called
\emph{weak-head normalization} and has the advantage that evaluation
never encounters a free variable so long as it starts with a closed term.
Only needing to deal with closed programs is a great boon for
implementers.  Beyond the added simplicity that comes from knowing
we won't run into a variable, in a $\beta$ reduction
$(\lambda x.v)~v' \reduct v[v'/x]$ the fact that $v'$ is closed
means that the substitution operation $[v'/x]$ need not
rename variables in $v$. More generally, weak-head normalization
on closed programs avoids the \emph{variable capture problem},
which is quite convenient for implementations.
In addition, programmers wanting to work with infinite data structures,
which can improve modularity by separating unbounded producers from
consumers \cite{Hughes:1989:WFP:63410.63411}, might be more inclined
to use a non-strict programming language like Haskell rather than ML.

%

For these reasons, \begin{inparaenum}[(1)]
\item call-by-name (or call-by-need) evaluation, \item weak-head normal forms and \item functional
extensionality \end{inparaenum} are all desirable properties to have.
However, the combination of all three is inconsistent, representing a trilemma,
so we can pick at most two.
Switching to call-by-value evaluation respects
extensionality while computing to weak-head normal forms.
But, sticking with call-by-name forces us to abandon one of
the other two. 
There is a fundamental tension between
evaluation to weak-head normal form, which always finishes when it reaches a lambda, and the 
$\eta$ axiom, which tells us that a lambda might not be done yet.
For example, the $\eta$ law says that $\lambda x.\Omega x$ is the same
as $\Omega$, where $\Omega$ is the non-terminating computation
$(\lambda y.y~y)(\lambda y.y~y)$.  Yet, $\lambda x.\Omega x$ is a weak-head normal form 
that is \emph{done} while $\Omega$ isn't.  Thus, if we want to use
weak-head normal forms as our stopping point, the $\eta$ law becomes suspect.
This puts us in a worrisome situation: $\eta$-equivalent programs might have
different termination behavior.  As such, we cannot use essential
properties, like our earlier example of the monad laws for state, for
reasoning about our programs without the risk of changing a program
that works into one that doesn't.
The root of our problem is that we combined extensionality with
effects, namely non-termination.  This is one example of the recurrent
tension that arises when we add effects to the call-by-name
$\lambda$-calculus.  For example, with printing as our effect, we
would encounter a similar problem when combining $\eta$ with evaluation to weak-head normal forms.
Evaluating the term
``$\texttt{print "hello"}; (\lambda y.y)$" to its weak-head normal form
would print the string
$\texttt{"hello"}$, while evaluating the $\eta$-expanded term
$\lambda x.(\texttt{print "hello"}; (\lambda y.y)) x$ would not.

Further, $\eta$ reduction even breaks confluence when the $\lambda$-calculus
is extended with control effects.
Recent
works~\cite{carraro:stackcalculus:LSFA,lambdamucons} suggest how to solve the
problem by adopting an alternative view of functions.
We bring new insight into this view:
through abstract machines we show how
the calling context, i.e. the elimination form of a lambda abstraction, can be given first-class
status and
lambda abstractions can then be seen as
\emph{pattern-matching} on the calling context.
We then utilize the view that pattern-matching is simply
syntactic sugar for projection operations; this 
suggests how to continue computing under a lambda abstraction.
We present and relate a series of operational semantics and abstract machines
for head reduction motivated by this insight into the nature of lambda abstractions.

After reviewing 
the small-step operational semantics, big-step operational
semantics and Krivine abstract machine for weak-head evaluation (Section \ref{smallbigAM}), we turn our investigation to
head reduction with the goal of providing the three different styles of semantics.
We start our exploration of head reduction with the Krivine abstract machine because it helps
us to think about the evaluation context as a first class object, and extend it with a construct
that \emph{names} these contexts. This brings out the
\emph{negative} nature of functions \cite{DBLP:conf/esop/DownenA14}.  Functions are not constructed but are \emph{de-constructors}; it is the contexts of functions which are 
constructed. To emphasize this view, functions are presented as pattern-matching on the calling
context, which naturally leads to a presentation of functions
that translates pattern-matching into projections;
analogously to the two ways tuples are treated in programming languages.
By utilizing this approach based on projections, we modify our Krivine machine with control
to continue evaluation instead of getting stuck on a top-level lambda (Section \ref{functionsaspatternmatching}).
Having gathered intuition on contexts, we focus on the control-free
version of this machine.  This leads  to our first
abstract machine for head reduction, and we utilize the syntactic correspondence  \cite{Biernacka-Danvy:TOCL07}
to derive an operational semantics for head reduction in the lambda-calculus
(Section \ref{smallbigamhead}). The obtained operational semantics
is, however, more complicated than would be desirable,
and so we define a simpler but equivalent operational semantics.
By once again applying the syntactic correspondence, we derive an abstract machine for
head reduction which is not based on projections and, by way of the functional correspondence \cite{Ager:2003:FCE:888251.888254},
we generate a big-step
semantics which is shown to be equivalent to Sestoft's  big step semantics for head reduction \cite{Sestoft:2002:DLC:860256.860276}.
Finally, we
conclude with a more efficient implementation of the projection based machine that
coalesces multiple projections into one (Section \ref{coalesced}).

\section{Weak-head Evaluation: Small and Big-step Operational Semantics and an Abstract Machine}
\label{smallbigAM}

\begin{figure}[t]
$$\begin{array}{c}
 E \in \textrm{EvaluationContexts} ::= \Box \mid E~ v  \\
 \\
  E[(\lambda x.v)~v'] \mapsto E[v[v'/x]] 
  \end{array}$$
   \caption{Small-step weak-head reduction}
  \label{fig:weakheadsmallstep}
\end{figure}

The semantics of a programming language can come in different flavors.
It can be given by creating a mapping from a syntactic domain of programs into
an abstract domain of mathematical structures (denotational semantics), or it can be given only
in terms of syntactic manipulations of programs (operational semantics). Operational semantics
can be further divided into \emph{small-step} or \emph{big-step}.  A small-step operational semantics
shows step-by-step how a program transitions to the final result.  A
big-step operational semantics, contrarily, only shows  the relation between a program and its final result with no intermediate steps shown, as in
an evaluation function. We first start in Figure \ref{fig:weakheadsmallstep} with a  small-step call-by-name semantics for $\lambda$-calculus, whose
terms are defined as follows:
\[ v \in \textrm{Terms} ::= x \mid v~v' \mid \lambda x.v \enspace . \]
The  \emph{evaluation context}, denoted by $E$, is simply a term with a hole, written
as $\Box$, which specifies where work occurs in a term.  A bare $\Box$
says that evaluation occurs at the top of the program, and if that is not
reducible then $E~v$ says that the search should continue
to the left of an application.
The semantics is then given by a single transition rule which specifies how to handle application.
According to this semantics, terms of the form
$x~ ((\lambda x.x)y)$ or $\lambda x. x~ ((\lambda x.x)y)$ are not reducible.
The final answer obtained is called  \emph{weak-head normal form} (whnf for short); a lambda abstraction is
in whnf and an application of the form $x~v_1\cdots v_n$ is in whnf.
We can give a grammar defining a whnf 
by using the notion of ``neutral'' from Girard for $\lambda$-calculus terms other
than lambda abstractions \cite{girard1989proofs} (see 
Figure \ref{fig:weakHeadNormalForms}). 

\begin{figure}[t]
  \begin{align*}
    \mbox{N} \in \textrm{Neutral} &::= x \mid \mbox{N}~ v\\
    \mbox{WHNF} \in \textrm{WeakHeadNormalForms} &::= \mbox{N} \mid \lambda x.v
  \end{align*}
  \caption{Weak-head normal forms}
  \label{fig:weakHeadNormalForms}
\end{figure}

\begin{figure}[t]
  {
    \[
    \begin{array}{l}
    \begin{array}{ll}
      c \in \textrm{Commands} &::= \cmd{v}{E} \\
         E \in \textrm{CoTerms} &::=  \tpcst \mid v \cdot E  \\
    \end{array} \\ 
   \\
    \begin{array}{ll}
      \cmd{v~v'}{E} &\reduct \cmd{v}{v' \cdot E} \\
      \cmd{\lambda x.v}{v' \cdot E} &\reduct \cmd{v[v'/x]}{E}
    \end{array}\end{array}\]
    \caption{Krivine abstract machine}
    \label{fig:krivinemachine}
  }
\end{figure}

Note that decomposing a program into an evaluation context and a redex 
is a meta-level operation in the small-step operational semantics.  We can instead make this operation an explicit part of the formalization in an \emph{abstract machine}.
In Figure~\ref{fig:krivinemachine} we give the Krivine abstract
machine~\cite{Krivine:2007:CLM:1325146.1325153}, which
can be derived directly from the operational semantics by reifying
the evaluation context  into a data structure called a \emph{co-term} \cite{Biernacka-Danvy:TOCL07}.
The  co-term $\tpcst$ is understood as corresponding to the empty context $\Box$,
while the call-stack $v \cdot E$ can be thought of as the context $E[\Box~v]$.
With this view, the formation of a command $\cmd{v}{E}$ corresponds to plugging $v$ into $E$ to obtain $E[v]$.
The reduction rules  of the Krivine machine are
justified by this correspondence: we have a rule that
recognizes that $E[v~v'] = (E[\Box~v'])[v]$ and
a rule for actually performing $\beta$ reduction inside an evaluation context.
We can further describe how to run a $\lambda$-calculus
term in the Krivine machine by plugging the term into  an empty context
\[ v \creduct \cmd{v}{\tpcst} \]
then after performing as many evaluation steps as possible, we ``readback'' a lambda term using the
 rules
\[   \cmd{v}{v' \cdot E} \hreduct \cmd{v~v'}{E}
\qquad\qquad\qquad\qquad
\cmd{v}{\tpcst}  \hreduct v 
\]
Note that the first rule is only needed if we want to interpret open programs, because execution only terminates
in commands of the form $\cmd{\lambda x.v}{\tpcst}$ and $\cmd{x}{E}$, and only the first of these can be closed.

The syntactic correspondence of Biernacka and Danvy \cite{Biernacka-Danvy:TOCL07} derives the Krivine machine from the small-step operational semantics of weak-head reduction by considering
them both as functional programs and applying a series of correctness-preserving program transformations.
The interpreter corresponding to the small-step semantics
``decomposes'' a term into an evaluation context and redex, reduces the redex, ``recompose'' the resulting term back into its context, and repeats this process until an answer is reached.
Recomposing and decomposing always happen in turns, and decomposing always undoes recomposing to arrive again at the same place,  so they can be merged into a single ``refocus'' function that
searches for the next redex in-place.  This non-tail recursive interpreter can then be made tail-recursive by inlining and fusing refocusing and reduction together.  Further simplifications and compression of intermediate transitions in the tail-recursive interpreter yields an implementation of the abstract machine.
Equivalence of the two semantic artifacts follows from the equivalence of the associated interpreters, which is guaranteed by construction due to the correctness of each program transformation used.
Note that we use $\reds$ and $\hreduct\!\!\!\!\!\!\to$ as the reflexive-transitive closures of $\reduct$ and $\hreduct$ respectively. 

\begin{theorem}[Equivalence of Krivine machine and small-step operational semantics]
\label{correspondence}
  For any $\lambda$-calculus terms $v,v'$ the following conditions are equivalent:
  \begin{enumerate}
  \item $v \dmapsto v'$ such that there is no $v''$ where $v' \mapsto v''$;
  \item there exists a command $c$ such that $\cmd{v}{\tpcst} \reds c \hreduct\!\!\!\!\!\!\to v'$
   where there is no $c'$ such that $c \reduct c'$.
%
  \end{enumerate}
\end{theorem}

\begin{figure}[t]
  \begin{mathpar}
    \inferrule{ }{x \Downarrow_{wh} x}~~~~
    \inferrule{ }{\lambda x.v \Downarrow_{wh} \lambda x.v}\\
    \inferrule{v_1 \Downarrow_{wh} \lambda x.v_1' \\ v_1'[v_2/x] \Downarrow_{wh} v}{v_1~v_2 \Downarrow_{wh} v}\\
    \inferrule{v_1 \Downarrow_{wh} v_1' \\ v_1' \not \equiv \lambda x.v}{v_1~v_2 \Downarrow_{wh} v_1'~v_2}
  \end{mathpar}
  \caption{Big-step semantics for weak-head reduction}
  \label{fig:weakheadbigstep}
\end{figure}

Let us now turn to the big-step weak-head semantics (see Figure \ref{fig:weakheadbigstep}).
It is not obvious that this semantics  corresponds to the small step semantics,
a proof of correctness is required.
Interestingly, Reynolds's functional correspondence 
\cite{Ager:2003:FCE:888251.888254, Reynolds:1972} 
links this semantics to the Krivine abstract machine by way of program transformations, similar to the connection between the small-step semantics and abstract machine.
The big-step semantics is represented as  a compositional interpreter which is then
converted into continuation-passing style and
defunctionalized (where higher-order  functions are replaced with data structures
which correspond to co-terms), yielding an interpreter representing the Krivine machine. Correctness follows
from construction and is expressed
analogously to Theorem \ref{correspondence} by replacing the small-step reduction (i.e.  $\dmapsto$) 
with the big-step (i.e. $\Downarrow$).

\begin{theorem}[Equivalence of Krivine machine and big-step operational semantics]
\label{bigcorrespondence}
  For any $\lambda$-calculus terms $v,v'$ the following conditions are equivalent:
  \begin{enumerate}
  \item $v  \Downarrow_{wh} v'$;
  \item there exists a command $c$ such that $\cmd{v}{\tpcst} \reds c \hreduct\!\!\!\!\!\!\to v'$
   where there is no $c'$ such that $c \reduct c'$.
    \end{enumerate}
\end{theorem}

In conclusion, we have seen three different semantics artifacts,  small-step, big-step and an
abstract machine, which thanks to the syntactic and functional correspondence 
define the \emph{same} language. 
Our goal is  to provide the three different styles of semantics for a different
notion of final result: \emph{head normal forms} (hnf for short) (see 
Figure \ref{fig:headNormalForms}).
Note that head reduction unlike weak-head reduction allows execution under a lambda abstraction, e.g.
 $\lambda x. (\lambda z.z)x$ is a whnf but is  not a hnf, whereas
$\lambda x.(x (\lambda z.z)x)$ is both a whnf and a hnf. 

\begin{figure}[t]
  \begin{align*}
      \mbox{N} \in \textrm{Neutral} &::= x \mid \mbox{N}~ v\\
    \mbox{HNF} \in \textrm{HeadNormalForms} &::= \mbox{N} \mid \lambda x.\mbox{HNF}
  \end{align*}
  \caption{Head Normal Forms}
  \label{fig:headNormalForms}
\end{figure}

\section{Functions as Pattern-Matching on the Calling Context}
\label{functionsaspatternmatching}
We have seen how, in the Krivine machine,
 evaluation contexts take the form of a call-stack consisting of a list
of arguments to the function being evaluated.
Inspired by this direct representation of contexts, we can enhance the language with the ability to give a name
to the calling context, analogous to naming terms
with a {\sf let}-construct. We introduce a new sort of variables (written using greek letters $\alpha,\beta,\ldots$), called \emph{co-variables},
which name co-terms, 
and a new abstraction $\mu \alpha . c$.
Operationally, a $\out$-term captures its evaluation context
\[ \cmd{\out \alpha.c}{E} \reduct c[E/\alpha] \]
and substitutes it in for the variable $\alpha$ in the associated command.
This is analogous to the evaluation of the term
$\texttt{let}~x = v ~ \texttt{in}~v'$, where $v$ is substituted for each occurrence of $x$
in $v'$. 
Interestingly, even though the $\mu$ rule seems very different from the
application rule they are indeed very similar, when seen side-by-side:
\[ \begin{array}{lll}
\cmd{\out \alpha.c}{E} & \reduct  & c[E/\alpha] \\
  \cmd{\lambda x.v}{v' \cdot E} &\reduct & \cmd{v[v'/x]}{E}
\end{array}
\]
Note that they both inspect  the context or co-term. Similar to the way a $\mu$-term corresponds to
a {\sf let}-term,
the lambda abstraction corresponds to a term of the form 
$\texttt{let}~ (x,y)= v ~ \texttt{in}~ v'$, which decomposes $v$ while naming its sub-parts.
So in contrast to the $\mu$-term,
the lambda abstraction decomposes the co-term instead of just naming it. To emphasize this view we write a function
as a special form of $\mu$-abstraction which  \emph{pattern-matches} on the context. The term
$\caseV[(x \cdot \alpha).c]$ gives names to both its argument $x$ and the remainder of its context $\alpha$ in the
command $c$. Operationally, $\caseV[(x \cdot \alpha).c]$
can be thought of as waiting for a context $v \cdot E$ at which point
computation continues in $c$ with $v$ substituted in for $x$ and $E$
substituted in for $\alpha$.
This view emphasizes that a function is given
primarily by how it is used rather than how it is defined; 
 the call-stack formation operator $\cdot$ is 
the most important aspect in the theory of functions (rather than
lambda abstraction).  
However, the two views of functions are equivalent.  We can write the lambda
abstraction $\lambda x.v$ as $\caseV[(x \cdot \alpha).\cmd{v}{\alpha}]$,
given $\alpha$ not free in $v$.  The application rule of the Krivine machine
can be clearly seen as an example of pattern-matching when written
in this style.
\[ \cmd{\lambda x.v}{v' \cdot E} = \cmd{\caseV[(x \cdot \alpha).\cmd{v}{\alpha}]}{v' \cdot E} \reduct \cmd{v[v'/x]}{E} \]
Similarly, we can also write $\caseV[(x \cdot \alpha).c]$ as $\lambda x.\out \alpha.c$,
making it clear that these two formulations are the same.  Interestingly, abstraction over co-terms is the only necessary ingredient to realizing control operations.
We thus arrive at an extension of the Krivine machine with control given in Figure \ref{fig:krivinemachinepluscontrol}.

\begin{figure}[t]
{
\[
\begin{array}{l}
\begin{array}{ll}
  c \in \textrm{Commands} &::= \cmd{v}{E} \\ 
  v \in \textrm{Terms} &::= x \mid v~v' \mid \caseV[(x \cdot \alpha).c] \mid \mu \alpha.c \\
  E \in \textrm{CoTerms} &::= \alpha \mid v \cdot E \mid  \tpcst   \\
\end{array} \\
\\
\begin{array}{lll}
  \cmd{v~v'}{E} & \reduct & \cmd{v}{v' \cdot E}\\
  \cmd{\out \alpha.c}{E} &\reduct & c[E/\alpha] \\
  \cmd{\caseV[(x \cdot \alpha).c]}{v \cdot E} & \reduct  & c[E/\alpha,v/x] 
\end{array}\end{array}
\]
}
\caption{Krivine abstract machine for $\lambda$-calculus with control}
\label{fig:krivinemachinepluscontrol}
\end{figure}

\subsection{Surjective Call Stacks}

There are two different ways to take apart tuples in programming languages.
The first, as we've seen, is to provide functionality to decompose a tuple by
matching on its structure, as in the pattern-matching {\sf let}-term $\texttt{let}~(x,y)=v'~\texttt{in}~v$.  By pattern-matching, $\texttt{let}~(x,y)=(v_1,v_2)~\texttt{in}~v$ evaluates to $v$
with $v_1$ and $v_2$ substituted for $x$ and $y$, respectively.
The second is to provide primitive projection operations for accessing the components
of the tuple, as in $\texttt{fst}(v)$ and $\texttt{snd}(v)$.
The operation $\texttt{fst}(v_1,v_2)$ evaluates to $v_1$ and $\texttt{snd}(v_1,v_2)$
evaluates to $v_2$.  These two different views on tuples are equivalent in a sense.
The $\texttt{fst}$ and $\texttt{snd}$ projections can be written in terms of pattern-matching
\[
\texttt{fst}(z) = (\texttt{let} ~ (x,y) = z ~ \texttt{in} ~ x)
\qquad\qquad
\texttt{snd}(z) = (\texttt{let} ~ (x,y) = z ~ \texttt{in} ~ y)
\]
and likewise, ``lazy'' pattern-matching can be implemented in terms of projection operations
\[
\texttt{let}~(x,y)=v'~\texttt{in}~v \reduct v[\texttt{fst}(v')/x,\texttt{snd}(v')/y]
\]
The projective view of tuples has the advantage of a simple interpretation of
extensionality, that the tuple made from the parts of another tuple is the same,
by the \emph{surjectivity} law for pairs:
$v=(\texttt{fst}(v),\texttt{snd}(v))$.

By viewing functions as pattern-matching constructs in a programming language,
analogous to pattern-matching on tuples,
we likewise have another interpretation of functions based on projection.
That is, we can replace lambda abstractions or pattern-matching with projection
operations, $\car{E}$ and $\cdr{E}$, for accessing the components of a calling context.
The operation $\car{v \cdot E}$ evaluates to the argument $v$ and $\cdr{v \cdot E}$
evaluates to the return context $E$.
Analogously to the different views on tuples, this projective view on functional contexts
can be used to implement pattern-matching:
\[
\cmd{\caseV[(x \cdot \alpha).c]}{E}
\reduct
c[\car{E}/x,\cdr{E}/\alpha]
\]
And since lambda abstractions can be written in terms of pattern-matching,
they can also be implemented in terms of $\car{-}$ and $\cdr{-}$:
\[
\cmd{\lambda x.v}{E}
=
\cmd{\caseV[(x \cdot \alpha).\cmd{v}{\alpha}]}{E}
\reduct
\cmd{v[\car{E}/x]}{\cdr{E}}
\]
This projective view of functions has been previously used in Nakazawa and Nagai's $\Lambda\mu_{\textrm{cons}}$-calculus~\cite{lambdamucons} to establish confluence in a call-by-name $\lambda$-calculus with control.
In this setting, extensional reasoning is captured by a surjectivity law on co-terms that
a co-term is always equal to the call-stack formed from its projections,
analogous to the law for surjective pairs:
$E = \car{E} \cdot \cdr{E}$.
Note that the equational soundness of the $\eta$-respecting translation of the pattern-matching lambda into projection 
has also been discovered in the context of the sequent calculus \cite{Her05,Munch13PhD}.
Therefore, the evaluation contexts of extensional functions are \emph{surjective call-stacks}.

In the case where the co-term is $v \cdot E$, the reduction of pattern-matching into projection
justifies our existing reduction rule by performing reduction inside of a command:
\[
\cmd{\caseV[(x \cdot \alpha).c]}{v \cdot E}
\reduct
c[\car{v \cdot E}/x,\cdr{v \cdot E}/\alpha]
\reds
c[v/x,E/\alpha]
\]
However, when working with abstract machines we want to keep reduction at the top of a program,
therefore we opt to keep
using the rule which combines both steps into one. 
Instead, rewriting lambda abstractions as projection
suggests what to do in the case where $E$ is \emph{not} a stack extension.
Specifically, in the case where $E$ is the top-level constant $\tpcst$ we have 
the following rule
\[ \cmd{\caseV[(x \cdot \alpha).c]}{\tpcst} \reduct c[\car{\tpcst}/x,\cdr{\tpcst}/\alpha] \]
which has no equivalent in the Krivine machine where the left-hand side of this reduction is stuck.

\begin{figure}[t]
  {
    \[\begin{array}{l}
    \begin{array}{ll}
      c \in \textrm{Commands} &::= \cmd{v}{E} \\ 
      v \in \textrm{Terms} &::= x \mid v~v' \mid \out \alpha.c \mid  \caseV[(x \cdot \alpha).c] \mid \car{S} \\
      E \in \textrm{CoTerms} &::= \alpha \mid v \cdot E \mid S  \\
      S \in \textrm{StuckCoTerms} &::= \tpcst \mid \cdr{S}\\
    \end{array}\\
    \\
    \begin{array}{lll}
    \cmd{v~v'}{E} & \reduct & \cmd{v}{v' \cdot E}\\
      \cmd{\out \alpha.c}{E} & \reduct  &c[E/\alpha] \\
      \cmd{\caseV[(x \cdot \alpha).c]}{v \cdot E} & \reduct & c[E/\alpha,v/x] \\
      \cmd{\caseV[(x \cdot \alpha).c]}{S} & \reduct  & c[\car{S}/x,\cdr{S}/\alpha]
    \end{array}
    \end{array}
    \]
  }
   \caption{Head reduction abstract machine for $\lambda$-calculus with control (projection based)}
  \label{fig:krivinemachinepluscontrolandprojections}
\end{figure}

We thus arrive at another abstract machine in Figure \ref{fig:krivinemachinepluscontrolandprojections} which works just like the Krivine machine with control except that now we
have a new syntactic sort of stuck co-terms.  The idea behind stuck co-terms is that $E$ is stuck if $\cdr{E}$ does not evaluate further.  For example, the co-term $\cdr{\cdr{\tpcst}}$ is stuck, but $\cdr{v \cdot \tpcst}$ is not.  Note, however, that we do not consider co-variables to be stuck co-terms since they may not continue
to be stuck after substitution.
For instance, $\cdr{\alpha}[v \cdot \tpcst/\alpha]=\cdr{v \cdot \tpcst}$ is not stuck, so
neither is $\cdr{\alpha}$.
This means that we have no possible reduction for the command $\cmd{\caseV[(x\cdot\beta).c]}{\alpha}$, but this is not
a problem since we assume to work only with programs which do not have free co-variables.
Further, because we syntactically restrict the use of the projection to stuck co-terms, we
do not need reduction rules like $\car{v \cdot E} \reduct v$ since $\car{v \cdot E}$ is not syntactically well formed in our machine.
Intuitively, anytime we would have generated $\car{v \cdot E}$ or $\cdr{v \cdot E}$ we instead eagerly perform the projection reduction in the other rules.

With the projection based approach, we are in a situation where there is always a reduction rule which can fire at the top of a command,
except for commands of the form $\cmd{x}{E}$ or $\cmd{\car{S}}{E}$, which are done,
or $\cmd{\caseV[(x \cdot \beta).c]}{\alpha}$ which is stuck on a free co-variable.
Specifically, we are no longer stuck when evaluating a pattern-matching function term
at the top-level,
i.e. $\cmd{\caseV[(x \cdot \alpha).c]}{\tpcst}$.
As a consequence of this, reduction of co-variable closed commands now respects $\eta$.
If we have a command of the form
$ \cmd{\caseV[(x\cdot \alpha).\cmd{v}{x \cdot \alpha}]}{E} $
with no free co-variables and where $x$ and $\alpha$ do not appear free in $v$, there are two possibilities depending on the value of $E$.
Either $E$ is a call-stack $v' \cdot E'$, so we $\beta$ reduce
\[ \cmd{\caseV[(x \cdot \alpha).\cmd{v}{x \cdot \alpha}]}{v' \cdot E'} \reduct \cmd{v}{v' \cdot E'} \]
or $E$ must be a stuck co-term, so we reduce by splitting it with projections
\[ \cmd{\caseV[(x \cdot \alpha).\cmd{v}{x \cdot \alpha}]}{S} \reduct \cmd{v}{\car{S} \cdot \cdr{S}} \]
meaning that we can continue to evaluate $v$.
Thus, the use of projection out of surjective call-stacks offers a way of implementing call-by-name reduction while also
respecting $\eta$.

\begin{figure}[t]
{
\[
\begin{gathered}
\begin{aligned}
  c \in \textrm{Command} &::= \cmd{v}{E} \\
  v \in \textrm{Terms} &::= x \mid v~v \mid \lambda x.v \mid \car{S} \\
  E \in \textrm{CoTerms} &::= v \cdot E \mid S \\
  S \in \textrm{StuckCoTerms} &::= \tpcst \mid \cdr{S}
\end{aligned} \\ \\
\begin{array}{lll}
  \cmd{v~v'}{E} &\reduct & \cmd{v}{v' \cdot E} \\
    \cmd{\lambda x.v}{v' \cdot E} & \reduct  & \cmd{v[v'/x]}{E} \\
  \cmd{\lambda x.v}{S} 
  &\reduct & \cmd{v[\car{S}/x]}{\cdr{S}}
\end{array}
\end{gathered}
\]
}
\caption{
Head reduction abstract machine for $\lambda$-calculus (projection based)} 
\label{fig:headreductionProjections}
\end{figure}

\section{Head Evaluation: Small and Big-step Operational Semantics and Abstract Machines} 
\label{smallbigamhead}

We have considered  the Krivine machine with control to get an intuition about dealing with
co-terms, however, our primary interest is to work with the pure $\lambda$-calculus.  Therefore,
from the Krivine machine with control and projection, we  derive an abstract machine for the pure $\lambda$-calculus which performs head reduction (see 
Figure \ref{fig:headreductionProjections}). Observe that the only co-variable
 is $\tpcst$, which represents the ``top-level'' of the program.
 Here we use $\car{-}$ and $\cdr{-}$ (as well as the top-level context $\tpcst$) as part
of the implementation since they can appear in intermediate states of the machine.  However, we still
assume that the programs being evaluated are pure lambda terms.
Observe that the projection based approach works just like the original Krivine machine (see Figure  \ref{fig:krivinemachine}), except in the case where we need to reduce
a lambda in a context that is not manifestly a call-stack ($\cmd{\lambda x.v}{S}$), and in that situation we
continue evaluating the body of the lambda.  To avoid the problem of having free variables,
we replace a variable $(x)$ with the projection into the context $(\car{S})$ and indicate that we
are now evaluating under the binder with the new context $(\cdr{S})$.

Not all commands derivable from the grammar of
Figure \ref{fig:headreductionProjections} are sensible; for example,
$\cmd{\car{\tpcst}}{\tpcst}$ and
$\cmd{x}{\car{\cdr{\tpcst}}.\cdr{\tpcst}}$ are not. Intuitively, the presence of $\car{\tpcst}$ means that
reduction has gone under a lambda abstraction so the co-term needs
to witness that fact by terminating in $\cdr{\tpcst}$, making $\cmd{\car{\tpcst}}{\cdr{\tpcst}}$ a legal command.
Analogously, the presence of $\car{\cdr{\tpcst}}$ indicates that reduction has gone under two
lambda abstractions and therefore the co-term needs
to terminate in $\cdr{\cdr{\tpcst}}$ making $\cmd{x}{\car{\cdr{\tpcst}}.\cdr{\cdr{\tpcst}}}$ a legal command.
To formally define the notion of a legal command, we make use of the notation $\cdnr{-}$ for $n$ applications of the
$\texttt{cdr}$ operation, and we write $\cdnr{\tpcst} \leq \cdmr{\tpcst}$ if $n \leq m$ and similarly $\cdnr{\tpcst} < \cdmr{\tpcst}$ if $n < m$.  We will only consider legal commands, and obviously reduction preserves legal commands.
\begin{definition}
A command of the form $\cmd{v_1}{v_2 \cdots v_n \cdot S}$ is legal if and only if
for $1 \leq i \leq n$ and for every $\car{S'}$ occurring  in $v_i$, $S' < S$.
\end{definition}


As we saw for the Krivine machine, we have an associated readback function with an additional rule
 \[  \cmd{v}{\cdr{S}} \hreduct \cmd{\lambda x.v[x/\car{S}]}{S} \]
for
extracting resulting $\lambda$-calculus terms after reduction has terminated, where
$v[x/\car{S}]$ is understood as replacing every occurrence of $\car{S}$ in $v$ with $x$ (which is assumed to be fresh).
The readback relation is justified by reversing the direction of the reduction
$ \cmd{\lambda x.v}{S} \reduct \cmd{v[\car{S}/x]}{\cdr{S}}$.
\begin{example}
\label{tracen}
If we start with the term $\lambda x.(\lambda y.y)~x$, we get an evaluation trace
\begin{align*}
\cmd{\lambda x.(\lambda y.y)~x}{\tpcst} & \reduct \cmd{(\lambda y.y)~\car{\tpcst}}{\cdr{\tpcst}} \\
& \reduct \cmd{\lambda y.y}{\car{\tpcst} \cdot \cdr{\tpcst}} \\
& \reduct \cmd{\car{\tpcst}}{\cdr{\tpcst}} \\
& \hreduct \cmd{\lambda x.x}{\tpcst} \\
& \hreduct \lambda x.x 
\end{align*}
which reduces $\lambda x.(\lambda y.y)~x$ to $\lambda x.x$.  
This corresponds to the intuitive idea that head reduction performs
weak-head reduction until it encounters a lambda, at which point it recursively
performs head reduction on the body of that lambda.
\end{example}
\begin{figure}[t]
  {
    \[
    \begin{gathered}
      \begin{aligned}
        t \in \textrm{TopTerms} &::= v \mid \lambda .t\\
        v \in \textrm{Terms} &::= x \mid v~v \mid \lambda x.v \mid i \\
        i \in \textrm{Indices} &:= \textrm{zero} \mid \textrm{succ}(i) \\
        E \in \textrm{Contexts} &::= E~v \mid \Box \\
        S \in \textrm{TopLevelContexts} &::= \lambda .S \mid \Box
      \end{aligned}
      \begin{aligned}
        Count & : \textrm{TopLevelContexts} \to \textrm{Indices}\\
        Count(\Box) &= \textrm{zero} \\
        Count(\lambda .S) &= \textrm{succ}(Count(S))
      \end{aligned}\\
      \\
      \begin{aligned}
        S[E[(\lambda x.v)~v']] &\mapsto  S[E[v[v'/x]]] \\
        S[\lambda x.v] &\mapsto  S[\lambda .v[Count(S)/x]]
      \end{aligned}
    \end{gathered}
    \]
  }
  \caption{Small-step head reduction corresponding to the  projection based abstract machine} 
  \label{fig:headreductionderived}
\end{figure}

Applying  Biernacka and Danvy's  syntactic correspondence, we reconstruct the small-step operational semantics of  Figure \ref{fig:headreductionderived}.
Because we want to treat contexts and terms separately, we create a new syntactic category of indices which
replaces the appearance of $\car{S}$ in terms, since stuck co-terms correspond to top-level contexts.
The function $Count(-)$ is used for converting between top-level contexts and indices. 
Note that, as we now include additional syntactic objects beyond the pure $\lambda$-calculus,
we need a readback relation just like in the abstract machine:
\[S[\lambda.v] \hreduct S[\lambda x.v[x/Count(S)]]\]

\begin{example}
Corresponding to the abstract machine execution  in Example \ref{tracen}  we have:
\begin{align*}
\lambda x.(\lambda y.y)~x& \mapsto  &  \mbox{ where } S =  \Box\\
\lambda .(\lambda y.y)~\textrm{zero}  & \mapsto & \mbox{ where } S = \lambda . \Box \mbox{ and } E = \Box\\
\lambda . \textrm{zero}  & \hreduct  \\
\lambda x.x 
\end{align*} 
Here, the context $\lambda .\Box~\textrm{zero}$ corresponds to the 
co-term $\car{\tpcst} \cdot \cdr{\tpcst}$. 
\end{example}

Because abstract machine co-terms correspond to inside out contexts, we can not just define evaluation contexts as
\[ E \in \textrm{Contexts} ::= E~v \mid S \]
which would make $((\lambda .S)~v)$ a context which
does not correspond to any co-term (and would allow for reduction under binders).
Indeed, the variable-free version of lambda abstraction is 
only allowed to occur at the top of the program.
Thus, composed contexts of the form $S[E[\Box]]$ serve as the exact equivalent of co-terms.
Analogously to the notion of legal commands, we have the notion of legal top-level terms, and we will only consider legal terms.

\begin{definition}
  A top-level term $t$ of the form $S[v]$   is legal if and only if for every index $i$ 
  occurring  in $v$, $i < Count(S)$.
\end{definition}

This operational semantics has the virtue of being reconstructed directly from
the abstract machine, which automatically gives a correctness result
analogous to Theorem \ref{correspondence}.

\begin{theorem}[Equivalence of small-step semantics and abstract machine based on projections]
  For any index free terms $v,v'$ the following conditions are equivalent:
  \begin{enumerate}
  \item $v \dmapsto t \hreduct\!\!\!\!\!\! \to v'$ such that there is no $t'$ where $t \mapsto t'$;
  \item there exists a command $c$ such that $\cmd{v}{\tpcst} \reds c \hreduct\!\!\!\!\!\! \to  v'$ where there is no $c'$ such that $c \red c'$.
  \end{enumerate}
\end{theorem}

However, while, in our opinion, the associated abstract machine was extremely elegant, this
operational semantics seems unnecessarily complicated.
Fortunately, we can slightly modify it to achieve a much simpler presentation.
At its core, the cause of the complexity of the operational semantics is the use of
indices for top-level lambdas.
\begin{figure}[t]
  {
    \[
    \begin{gathered}
      \begin{aligned}
        v \in \textrm{Terms} &::= x \mid v~v \mid \lambda x.v \\
        E \in \textrm{Contexts} &::= E~v \mid \Box \\
        S \in \textrm{TopLevelContexts} &::= E \mid \lambda x.S
      \end{aligned} \\\\
      S[(\lambda x.v)~v'] \mapsto S[v[v'/x]]
    \end{gathered}
    \]
  }
\caption{Small-step head reduction}
\label{fig:headreduction}
\end{figure}
A simpler operational semantics would only use named variables (see
Figure \ref{fig:headreduction}). To show that the two semantics are indeed equivalent we
make use of  Sabry and Wadler's  \emph{reduction correspondence} \cite{DBLP:journals/toplas/SabryW97}.  For readability, we
write $\mapsto_S$ for the evaluation according to Figure \ref{fig:headreductionderived}
and $\mapsto_T$ for the evaluation according to Figure  \ref{fig:headreduction}.
We define the translations
$*$ and $\#$ from top-level terms, possibly containing indices, to pure lambda-calculus terms, and from
pure lambda terms to top-level terms, respectively. More specifically,
$*$ corresponds to the normal form of the readback rule and $\#$ to the normal form
with respect to non-$\beta$ $\mapsto_S$ reductions.

\begin{theorem}
The reduction systems $\dmapsto_S$ and $\dmapsto_T$ are sound and complete with respect to each other:
\begin{enumerate}
\item
$t \dmapsto_S t'$ implies $t^* \dmapsto_T t'^*$;
\item
$  v\dmapsto_T v'$ implies $v^\#  \dmapsto_S v'^\#$;
\item
$t \dmapsto_S t^{*\#}$ and $v^{\#*} = v$.
\end{enumerate}
\end{theorem}
\begin{proof}
  First, we observe that the non-$\beta$ reduction of $\mapsto_S$ is inverse to
  the readback rule, so that for any top-level terms $t$ and $t'$,
  $t \mapsto_S t'$ by a non-$\beta$ reduction if and only if $t' \hreduct t$.
  Second, we note that for any top-level context
  $S=\lambda\dots\lambda.\Box$ and term $v$ of Figure
  \ref{fig:headreductionderived},
  there is a top-level context $S'=\lambda x_0\dots\lambda x_n.\Box$
  of Figure \ref{fig:headreduction} and substitution
  $\sigma=[x_0/0,\dots,x_n/n]$ such that
  $(S[v])^*=S'[v^\sigma]$ by induction on $S$. Similarly, for any top-level
  context $S$ of Figure \ref{fig:headreduction} and neutral term $N$, there is a
  top-level context
  $S'$ and context $E'$ of Figure \ref{fig:headreductionderived} such that
  $(S[N])^\#=S'[E'[N^\sigma]]$ by induction on $S$.

\begin{enumerate}
\item
Follows from the fact that each step of $\mapsto_S$ corresponds to zero or one
step of $\mapsto_T$:
\[
\begin{array}{ll}
  \begin{tikzcd}
    S[E[(\lambda x.v)v']] \arrow[|->]{r}{S} \arrow{d}{*}
    &
    S[E[v[v'/x]]] \arrow{d}{*}
    \\
    S'[E^\sigma[(\lambda x.v^\sigma)v'^\sigma]] \arrow[|->,dashed]{r}{T}
    &
    S'[E^\sigma[v^\sigma[v'^\sigma/x]]]
  \end{tikzcd}
  & 
  \begin{tikzcd}
    S[\lambda x.v] \arrow[|->]{r}{S} \arrow{d}{*}
    &
    S[\lambda .v[Count(S)/x]]  \arrow{d}{*}
    \\
    S'[\lambda x. v^\sigma]  \arrow[equals]{r}{T}
    &
    S'[\lambda x. v^\sigma]
  \end{tikzcd} 
\end{array}
\]
\item Follows from the fact that each step of $\mapsto_T$ corresponds to one
  step of $\mapsto_S$, similar to the above.
\item Follows from induction on the reductions of the $*$ and $\#$
  translations along with the facts that the non-$\beta$ reduction of
  $\mapsto_S$ is
  inverse to $\hreduct$, $t^{*\#}$ is the normal form of $t$ with
  respect to non-$\beta$ reductions of $\mapsto_S$, and the top-level term $v$ is
  the normal form of the readback.
  \qedhere
\end{enumerate}
\end{proof}

\begin{figure}[t]
   {
  \[
  \begin{gathered}
    \begin{aligned}
      c \in \textrm{Command} &::= \cmd{v}{E} \\
      v \in \textrm{Terms} &::= x \mid v~v \mid \lambda x.v \\
      E \in \textrm{CoTerms} &::= v \cdot E \mid S \\
      S \in \textrm{StuckCoTerms} &::= \tpcst \mid \texttt{Abs}(x,S)
    \end{aligned} \\ \\
    \begin{array}{lll}
      \cmd{v~v'}{E} & \reduct & \cmd{v}{v' \cdot E} \\ 
      \cmd{\lambda x.v}{v' \cdot E} & \reduct & \cmd{v[v'/x]}{E} \\
      \cmd{\lambda x.v}{S} &\reduct &  \cmd{v}{\texttt{Abs}(x,S)}
    \end{array}
  \end{gathered}
  \]
}
\caption{Head reduction abstract machine}
\label{fig:variableMachine}
\end{figure}

\begin{remark}
The small-step semantics of Figure  \ref{fig:headreduction} captures the definition of head reduction
given by Barendregt (Definition 8.3.10) \cite{Barendregt} who defines a head redex as follows:
If $M$ is of the form $$\lambda x_1 \cdots x_n. (\lambda x. M_0)M_1 \cdots M_m \enspace , $$
for $n \geq 0$, $m \geq 1$, then $ (\lambda x. M_0)M_1 $ is called the
\emph{head redex} of $M$. Note that the context $\lambda x_1 \cdots x_n. \Box \cdots M_m$ is broken down into
 $\lambda x_1 \cdots x_n. \Box$ which corresponds to a top-level context $S$ and
$\Box \cdots M_n$ which corresponds to a context $E$.  Thus, Barendregt's notion of one-step head reduction:
$$ M \rightarrow_h  N  
\mbox{ if } M \rightarrow^{\Delta} N \enspace, $$
i.e. $N$ results from $M$ by contracting the head redex $\Delta$,
corresponds to the evaluation rule of Figure \ref{fig:headreduction}.
\end{remark}

By once again applying the syntactic correspondence, this time to the simplified operational semantics in Figure \ref{fig:headreduction}, we derive a new abstract machine (Figure \ref{fig:variableMachine}) which works by going under lambdas without performing any substitutions in the process.
To extract computed $\lambda$-calculus terms we add one rule to the readback relation for the Krivine machine.
\[  \cmd{v}{\texttt{Abs}(x,S)} \hreduct \cmd{\lambda x.v}{S} \]
The equivalence of the abstract machine in Figure \ref{fig:variableMachine} and the operational semantics of Figure \ref{fig:headreduction}
follows by construction, and is expressed analogously to Theorem \ref{correspondence}.


\begin{figure}[t]
  \begin{mathpar}
    \inferrule{v \Downarrow_{wh} \lambda x.v' \\ v' \Downarrow_h v''}{v \Downarrow_h \lambda x.v''}~~~~
    \inferrule{v \Downarrow_{wh} v' \\ v' \not \equiv \lambda x.v''}{v \Downarrow_h v'}
  \end{mathpar}
  \caption{Big-step semantics for head reduction}
  \label{fig:bigstepSemantics}
\end{figure}

\begin{figure}[t]
  \begin{mathpar}
    \inferrule{ }{x \Downarrow_{sf} x}~~~~
    \inferrule{v \Downarrow_{sf} v'}{\lambda x.v \Downarrow_{sf} \lambda x.v'}\\
    \inferrule{v_1 \Downarrow_{wh} v_1' \\ v_1' \not \equiv (\lambda x.v)}{v_1~v_2 \Downarrow_{sf} v_1'~v_2}~~~~~
    \inferrule{v_1 \Downarrow_{wh} \lambda x.v_1' \\ v_1'[v_2/x] \Downarrow_{sf} v}{v_1~v_2 \Downarrow_{sf} v}
  \end{mathpar}
  \caption{Sestoft's big-step semantics for head reduction}
  \label{fig:bigstepSemanticsSestoft}
\end{figure}
From the abstract machine in Figure \ref{fig:variableMachine} we again apply the functional correspondence
\cite{Ager:2003:FCE:888251.888254} to derive the big-step semantics in Figure \ref{fig:bigstepSemantics}.
This semantics utilizes the big-step semantics for weak-head reduction ($\Downarrow_{wh}$)
from Figure~\ref{fig:weakheadbigstep} as part of its definition of
head reduction ($\Downarrow_h$).
This  semantics tells us that the way to evaluate a term to head-normal form  is to first evaluate it to weak-head normal form,
and if the resulting term is a lambda to recursively evaluate the body of that lambda.
This corresponds to 
the behavior of the abstract machine which works just like the Krivine machine (which only reduces to weak head)
except in the situation where we have a command consisting of a lambda abstraction and a context which
is not a call-stack. Again, the big-step semantics corresponds to the abstract machine by construction,
the equivalence can be expressed analogously to Theorem \ref{bigcorrespondence}.

Interestingly,  Sestoft gave a different 
 big-step semantics for head reduction \cite{Sestoft:2002:DLC:860256.860276} (Figure \ref{fig:bigstepSemanticsSestoft}).
The only difference between the two semantics is in the way
they search for $\beta$ redexes.  Sestoft's semantics is more redundant by repeating
the same logic for function application from the underlying weak-head
evaluation, whereas the semantics of  Figure  \ref{fig:bigstepSemantics} only adds a loop for descending under
the top-level lambdas produced by weak-head evaluation. However, 
 besides this difference in the search for $\beta$ redexes,
they are the same: they eventually find and perform
$\beta$ redexes in exactly the same order.
By structural induction one can indeed verify that they generate identical
operational semantics.
%
\begin{theorem}
  $v \Downarrow_h v'$  if and only if $v \Downarrow_{sf} v'$.
\end{theorem}

\section{Coalesced Projection}
\label{coalesced}
The projection based machine in Figure \ref{fig:headreductionProjections} has the desirable
feature that we will never encounter a variable when we start with a closed program. Additionally, unlike the
machine in Figure \ref{fig:variableMachine}, machine states do not have co-terms that bind variables in their opposing term.
%
 On the other hand, it has a certain undesirable
property in that when we substitute a projection in for a variable
\[ \cmd{\lambda x.v}{S} \reduct \cmd{v[\car{S}/x]}{\cdr{S}} \]
we may significantly increase the size of the term as the stuck co-term
has size linear in the number of lambda abstractions we have previously
eliminated in this way.
The story is even worse in the other direction: the readback relation
for the projection machine (and associated operational semantics)
depends on performing a deep
pattern-match to replace projections with variables
\[ \cmd{v}{\cdr{S}} \hreduct \cmd{\lambda x.v[x/\car{S}]}{S} \]
which requires matching \emph{all the way} against $S.$

We now show that we can improve the abstract machine for
head evaluation by combining  multiple projections into one. 
We will utilize the macro projection operations $\pickE{n}{-}$ and
$\dropE{n}{-}$ which, from the perspective of $\car{-}$ and $\cdr{-}$,
coalesce sequences of projections into a single operation:
{
\begin{align*}
  \dropE{0}{E} & \triangleq E
  &
  \pickE{n}{E} & \triangleq \car{\dropE{n}{E}}
  \\
  \dropE{n+1}{E} & \triangleq \cdr{\dropE{n}{E}}
\end{align*}
}%
Given that our operational semantics is for the pure lambda calculus,
$\pickE{n}{\tpcst}$ and
$\dropE{n}{\tpcst}$ are the only call-stack projections we need in the
syntax.

We now construct an abstract machine (Figure
\ref{fig:headreductionMachine}) for head evaluation of the
$\lambda$-calculus which is exactly like the Krivine machine
with one additional rule utilizing the coalesced projections. As in the machine
of Figure \ref{fig:headreductionProjections}, the coalesced machine
 ``splits'' the top-level into a call-stack and
continues.  Analogously to Example  \ref{tracen}, one has:
{%
\begin{align*}
  \cmd{\lambda x. (\lambda y. y) x}{\dropE{0}{\tpcst}}
  &\reduct
  \cmd{(\lambda y. y) (\pickE{0}{\tpcst})}{\dropE{1}{\tpcst}}
  \\
  &\reds
  \cmd{\pickE{0}{\tpcst}}{\dropE{1}{\tpcst}}
\end{align*}
}%
\begin{figure}[t]
{
\[
\begin{gathered}
\begin{aligned}
  c \in \textrm{Command} &::= \cmd{v}{E} \\
  v \in \textrm{Terms} &::= x \mid v~v \mid \lambda x.v \mid \pickE{n}{\tpcst}\\
  E \in \textrm{CoTerms} &::= v \cdot E \mid \dropE{n}{\tpcst}
\end{aligned} \\ \\
\begin{array}{lll}
  \cmd{v~v'}{E} &\reduct & \cmd{v}{v' \cdot E} \\
  \cmd{\lambda x.v}{v' \cdot E} &  \reduct & \cmd{v[v'/x]}{E} \\
  \cmd{\lambda x.v}{\dropE{n}{\tpcst}}
  &\reduct &
  \cmd{v[\pickE{n}{\tpcst}/x]}{\dropE{n+1}{\tpcst}}
\end{array}
\end{gathered}
\]
}
\caption{Coalesced projection abstract machine }
\label{fig:headreductionMachine}
\end{figure}
If the machine
terminates with a result like $\cmd{\pickE{n}{\tpcst}}{E}$, it is
straightforward to read back the corresponding head normal form:
{ \[
\begin{gathered}
\begin{aligned}
 \cmd{v}{v' \cdot E} &\hreduct \cmd{v~v'}{E} \qquad  & \cmd{v}{\tpcst} & \hreduct v  
 &  \qquad 
 \cmd{v}{\dropE{n+1}{\tpcst}}
 &\hreduct
 \cmd{\lambda x.v[x/\pickE{n}{\tpcst}]}{\dropE{n}{\tpcst}}
\end{aligned}\\
\end{gathered}
\] }%
where $x$ is not free in $v$ in the third rule.  Note that the
substitution $v[x/\pickE{n}{\tpcst}]$ replaces all occurrences of
terms of the form $\pickE{n}{\tpcst}$ inside $v$ with $x$.
Correctness is ensured by the fact that reduction in the machine
with coalesced projections corresponds to reduction in the machine
with non-coalesced projections, which in turn corresponds to the operational
semantics of head reduction.

\begin{theorem}
[Equivalence of Coalesced and Non Coalesced Projection Machines]
$\phantom{}$

\begin{enumerate}
\item $c \reduct c'$ in the coalesced machine if and only if $c \reduct c'$ in the non-coalesced machine and
\item $c \hreduct c'$ in the coalesced machine if and only if $c \hreduct c'$ in the non-coalesced machine
\end{enumerate}
where conversion is achieved by interpreting $\dropE{n}{-}$ and $\pickE{n}{-}$ as macros.
\end{theorem}
\begin{proof}
By cases.  Note that the interesting case in each direction is $\cmd{\lambda x.v[x/\pickE{n}{\tpcst}]}{\dropE{n}{\tpcst}} \reduct \cmd{v}{\dropE{n+1}{\tpcst}}$ which follows since 
\begin{align*}
\cmd{\lambda x.v[x/\pickE{n}{\tpcst}]}{\dropE{n}{\tpcst}} &\reduct \cmd{v[x/\pickE{n}{\tpcst}][\car{\dropE{n}{\tpcst}}/x]}{\cdr{\dropE{n}{\tpcst}}} \\
&= \cmd{v[x/\pickE{n}{\tpcst}][\pickE{n}{\tpcst}]}{\dropE{n+1}{\tpcst}} \\
&= \cmd{v}{\dropE{n+1}{\tpcst}} \qedhere
\end{align*}
\end{proof}

Our abstract machine, in replacing variables with coalesced sequences of projections,
exhibits a striking resemblance to implementations of the $\lambda$-calculus based on de Bruijn indices \cite{deBruijn1972381}.
Indeed, our abstract machine can be seen as given a semantic justification for de Bruijn indices as offsets into the call-stack.
However, our approach differs from de Bruijn indices in that, in general, we still utilize named variables.
Specifically, numerical indices are only ever used for
representing variables bound in the leftmost branch of the lambda term viewed as a tree.
As such, we avoid the complexities of renumbering during $\beta$ reduction.
However, implementations which do use de Bruijn to achieve a nameless implementation of variables
in general have the advantage that the substitution operations $[\pickE{n}{\tpcst}/x]$ as used in
the abstract machine could be replaced with a no-op.
The Krivine machine is often presented using de Bruijn despite being used only for computing
whnfs.  With the addition of our extra rule---which in a de Bruijn setting does not require substitution---we extend it to compute hnfs.

Further, we can extract from our abstract machine an operational semantics which utilizes de Bruijn indices
only for top-level lambdas by way of the syntactic correspondence. 
\begin{figure}[t]
  {
\[
\begin{gathered}
  \begin{aligned}
  t \in \textrm{TopTerms} &::= \lambda^n.v\\
  v \in \textrm{Terms} &::= x \mid n \mid v~v \mid \lambda x.v \\
  E \in \textrm{Contexts} &::= \Box \mid E~v
  \end{aligned} \\
  \begin{array}{lll}
    \lambda^n.\lambda x.v &  \mapsto &  \lambda^{n+1}.v[n/x]\\
    \lambda^n.E[(\lambda x.v)~v']  & \mapsto  & \lambda^n.E[v[v'/x]]
  \end{array}
\end{gathered}
\]
}
\caption{Operational semantics of head reduction with partial de Bruijn}
\label{fig:headreductiondebruij}
\end{figure}
The resulting semantics, in Figure \ref{fig:headreductiondebruij}, keeps track of the number
of top-level lambdas as part of $\lambda^n.v$ while performing head reduction on $v$. As
expected, the equivalence between the coalesced abstract machine and the operational semantics
is by construction.
%
%
%
The coalesced projection machine, and de Bruijn based operational semantics, are essentially
equivalent to the projection machine and associated operational semantics.
The difference is that by coalescing multiple projections into one, we replace
the use of unary natural numbers with abstract numbers which could be implemented efficiently.
In this way, we both greatly reduce the size of terms and make the pattern-matching to readback
final results efficient.  
Numerical de Bruijn indices are an efficient implementation
of the idea that lambda abstractions use variables to denote projections out of a call stack.

\section{Conclusion}
The desirable combination of weak-head normal forms, call-by-name evaluation, and extensionality is not achievable.
This is a fundamental, albeit easily forgettable, constraint in programming language design.
However, we have seen that there is a simple way out of this trilemma: replace weak-head normal forms with head normal forms.
Moreover, reducing to head normal form is motivated by an analysis of the meaning of lambda abstraction.
We took a detour through control so that we could directly reason about not just terms but also their context.
This detour taught us to think about the traditional $\beta$ rule as a rule for pattern-matching on contexts.
By analogy to the different ways we can destruct terms we recognized 
an alternative implementation based on projections out of call-stacks.
Projection allows us to run a lambda abstraction even when we don't yet have its argument.

Returning to the pure $\lambda$-calculus we derived an abstract machine from this projection-based approach.
Using the syntactic correspondence we derived from our abstract machine an operational semantics, and showed how that
could be massaged into a simpler operational semantics for head reduction.
With a second use of the syntactic correspondence
we derived an abstract machine which implemented head reduction in what seems a more traditional way.
By the functional correspondence
we showed finally that our entire chain of abstract machines and operational semantics correspond to Sestoft's big-step semantics
for head reduction.
The use of automated techniques for deriving one form of semantics from another
makes it easy to rapidly explore the ramifications that follow from a shift of strategy;
in our case from weak-head to head reduction.
So in the end, we arrive at a variety of different semantics for head reduction,
all of which are equivalent and correspond to Barendregt's definition
of head reduction for the $\lambda$-calculus.

Branching from our projection based machine we derived an efficient abstract machine which coalesces projections,
so we may have our cake and eat it too.
We escape the trilemma by giving up on weak-head reduction, while still retaining its virtues like avoiding the variable capture problem
and keeping all reductions at the top-level.
Our machine is identical to the Krivine machine, which performs weak-head evaluation, except
for a single extra rule which tells us what to do when we have a lambda abstraction and no arguments to feed it.
Further, these changes can be seen as a limited usage of (and application for) de Bruijn indices.

Although our motivating problem and eventual solution were in the context of the pure $\lambda$-calculus,
our detour through explicit continuations taught us valuable lessons about operational semantics.
Having continuations forces us to tackle many operational issues head on because it makes the
contexts a first class part of the language.
Thus, a meta lesson of this paper is that adding control can be a helpful aid to designing even pure languages.

In practice, the trilemma is usually avoided by giving up on call-by-name or extensionality rather than
 weak-head normal forms.
However, it may be that effective approaches to head evaluation such as those in this paper are of interest even in
call-by-value languages or in settings (such as Haskell with \texttt{seq}) that lack the $\eta$ axiom.
Exploring the practical utility of head reduction may be a profitable avenue for future work.

\begin{figure}[t]
{\small
\[
\begin{aligned}
  c \in \textrm{Commands} &::= \cmdE{v}{\sigma}{E} \\
  v \in \textrm{Terms} &::= x \mid \lambda x.v \mid v~v \\
  E \in \textrm{CoTerms} &::=  \tpcst \mid t \cdot E \\
  \sigma \in \textrm{Environments} &::= [] \mid (x \mapsto t)::\sigma \\
  t \in \textrm{Closures} &::= (v,\sigma)
\end{aligned}
\]
\[
\begin{aligned}
  \cmdE{v~v'}{\sigma}{E} &\reduct \cmdE{v}{\sigma}{(v',\sigma) \cdot E} \\
  \cmdE{\lambda x.v}{\sigma}{t \cdot E} &\reduct \cmdE{v}{(x \mapsto t)::\sigma}{E} \\
  \cmdE{x}{\sigma}{E} &\reduct \cmdE{v}{\sigma'}{E}~~~\textrm{where $\sigma(x) = (v,\sigma')$}
\end{aligned}
\]
}
\caption{Krivine abstract machine with environments}
\label{fig:krivinemachineclosures}
\end{figure}
Finally, we presented our abstract machines using substitutions and have suggested
a possible alternative implementation based on de Bruijn indices, but we can also perform closure conversion to
handle variables. For example, the machine in Figure \ref{fig:krivinemachineclosures} is an implementation
of the Krivine machine using closures.
We assume the existence of a data structure for maps from variable names to closures which
supports at least an extension operator $::$ and variable lookup, which we write using list-like notation.
\begin{figure}[t]
{\small
\[
\begin{aligned}
  c \in \textrm{Commands} &::= \cmdE{v}{\sigma}{E} \\
  v \in \textrm{Terms} &::= x \mid \lambda x.v \mid v~v \mid \pickE{n}{\tpcst} \\
  E \in \textrm{CoTerms} &::=  \dropE{n}{\tpcst} \mid t \cdot E \\
  \sigma \in \textrm{Environments} &::= [] \mid (x \mapsto t)::\sigma \\
  t \in \textrm{Closures} &::= (v,\sigma)
\end{aligned}
\]
\[
\begin{aligned}
  \cmdE{v~v'}{\sigma}{E} &\reduct \cmdE{v}{\sigma}{(v',\sigma) \cdot E} \\
  \cmdE{\lambda x.v}{\sigma}{t \cdot E} &\reduct \cmdE{v}{(x \mapsto t)::\sigma}{E} \\
  \cmdE{\lambda x.v}{\sigma}{\dropE{n}{\tpcst}} &\reduct \cmdE{v}{(x \mapsto (\pickE{n}{\tpcst},\sigma))::\sigma}{\dropE{n+1}{\tpcst}} \\
  \cmdE{x}{\sigma}{E} &\reduct \cmdE{v}{\sigma'}{E}~~~\textrm{where $\sigma(x) = (v,\sigma')$}
\end{aligned}
\]
}
\caption{Head reduction abstract machine with environments}
\label{fig:headmachineclosures}
\end{figure}
Similarly, the machine in Figure \ref{fig:headmachineclosures} takes our efficient coalesced machine and replaces
the use of substitutions with environments.  
This would correspond to a calculus of explicit substitutions ˆ la $\lambda \rho$ \cite{curien1988lambda}.
However, in general, the problem of how to handle variables is orthogonal from questions of operational semantics,
and thus environment based handling of variables can be added after the fact.  What the present paper suggests,
however, is that the de Bruijn view of representing variables as numbers is semantically motivated
by the desire to make reduction for the call-by-name lambda calculus consistent with extensional principles.
\bibliographystyle{eptcs}
\bibliography{extens}

\end{document}